\newif\ifniceformat
\niceformattrue     

\ifniceformat
\documentclass[11pt]{article}
\usepackage{textcomp}
\usepackage{amsmath}
\usepackage{amssymb}
\usepackage{amsthm}
\usepackage{todonotes}
\usepackage{tikz}
\usepackage{pgfplots}
\usepackage{graphicx}
\usepackage{epstopdf}
\usepackage{wrapfig}
\usepackage[labelfont=bf,font=footnotesize]{caption}
\usepackage[labelfont=bf,font=footnotesize]{subcaption}
\usepackage{xcolor}
\usepackage{bbm}
\usepackage{makecell}
\usepackage{authblk}
\usepackage{fullpage,etoolbox}
\usepackage[shortlabels]{enumitem}
\usepackage[round,sort,compress]{natbib}
\usepackage{listings}
\usepackage{commands}
\usepackage[ruled,vlined,linesnumbered]{algorithm2e}
\usepackage{todonotes}
\usepackage{float}
\usepackage{makecell}
\usepackage{balance}  
\usepackage[colorlinks=true,citecolor=blue,linkcolor=blue,urlcolor=blue]{hyperref}

\usepackage{auxFiles/symbolDef}

\makeatletter
    \@ifpackageloaded{xcolor}{}{\usepackage{xcolor}}
\makeatother
\newcommand{\rewrite}[2]{{#2}}

\makeatletter
    \@ifpackageloaded{needspace}{}{\usepackage{needspace}}
\makeatother

\newcommand{\uppercaseGreek}[1]{%
    \begingroup
    \ucmathlist\MakeUppercase{#1}%
    \endgroup
}

\newcommand{\ucmathlist}{%
    \def\alpha{\mathrm{A}}%
    \def\beta{\mathrm{B}}%
    \let\gamma=\Gamma
    \let\delta=\Delta
    \def\epsilon{\mathrm{E}}%
    \def\varepsilon{\mathrm{E}}%
    \def\zeta{\mathrm{Z}}%
    \def\eta{\mathrm{H}}%
    \let\theta=\Theta
    \let\vartheta=\Theta
    \def\iota{\mathrm{I}}%
    \def\kappa{\mathrm{K}}%
    \let\lambda=\Lambda
    \def\mu{\mathrm{M}}%
    \def\nu{\mathrm{N}}%
    \let\xi=\Xi
    \let\pi=\Pi
    \let\varpi=\Pi
    \def\rho{\mathrm{P}}%
    \def\varrho{\mathrm{P}}%
    \let\sigma=\Sigma
    \def\tau{\mathrm{T}}%
    \let\upsilon=\Upsilon
    \let\phi=\Phi
    \let\varphi=\Phi
    \def\chi{\mathrm{X}}%
    \let\psi=\Psi
    \let\omega=\Omega
}

\makeatletter
    \@ifpackageloaded{amsthm}{}{

        \usepackage{amsthm}
    }
\makeatother
\theoremstyle{plain}
    \newtheorem{theorem}{Theorem}
    
    \newtheorem{lemma}[theorem]{Lemma}
    
\theoremstyle{definition}
    \newtheorem{definition}{Definition}
    \newtheorem{remark}{Remark}
    \newtheorem{assumption}{Assumption}
    \newtheorem*{example}{Example}

\makeatletter
\def\renewtheorem#1{%
    \expandafter\let\csname#1\endcsname\relax
    \expandafter\let\csname c@#1\endcsname\relax
    \gdef\renewtheorem@envname{#1}
    \renewtheorem@secpar
}
\def\renewtheorem@secpar{\@ifnextchar[{\renewtheorem@numberedlike}{\renewtheorem@nonumberedlike}}
\def\renewtheorem@numberedlike[#1]#2{\newtheorem{\renewtheorem@envname}[#1]{#2}}
\def\renewtheorem@nonumberedlike#1{  
    \def\renewtheorem@caption{#1}
    \edef\renewtheorem@nowithin{\noexpand\newtheorem{\renewtheorem@envname}{\renewtheorem@caption}}
    \renewtheorem@thirdpar
}
\def\renewtheorem@thirdpar{\@ifnextchar[{\renewtheorem@within}{\renewtheorem@nowithin}}
\def\renewtheorem@within[#1]{\renewtheorem@nowithin[#1]}
\makeatother

\else

\documentclass[letterpaper, 10 pt, conference]{ieeeconf}
\IEEEoverridecommandlockouts 
\usepackage[sort,compress]{cite}
\usepackage{graphicx}
\usepackage{xcolor}
\usepackage{amsmath, amssymb}
\usepackage{commands}
\usepackage{csquotes}
\usepackage[ruled,vlined,linesnumbered]{algorithm2e}
\usepackage{float}
\usepackage{todonotes}
\usepackage{subcaption}
\usepackage{makecell}
\usepackage[normalem]{ulem}
\newcommand{\citep}[1]{\cite{#1}}
\newcommand{\citet}[1]{\cite{#1}}

\usepackage{auxFiles/symbolDef}

\makeatletter
\let\NAT@parse\undefined
\makeatother
\usepackage[colorlinks=true, allcolors=blue]{hyperref}

\fi

\ifniceformat
\title{Learning Flatness-Preserving Residuals for Pure-Feedback Systems}
\else
\title{\LARGE \bf Learning Flatness-Preserving Residuals for Pure-Feedback Systems}
\fi

\author{Fengjun Yang, Jake Welde, and Nikolai Matni
\thanks{F. Yang, J. Welde, and N. Matni are with the GRASP Laboratory, University of Pennsylvania, PA, USA. This work was supported in part by NSF Awards SLES-2331880, ECCS-2045834, ECCS-2231349, and AFOSR Award FA9550-24-1-0102.}%
}

\begin{document}
\maketitle

\begin{abstract}%
We study residual dynamics learning for differentially flat systems, where a nominal model is augmented with a learned correction term from data. A key challenge is that generic residual parameterizations may destroy flatness, limiting the applicability of flatness-based planning and control methods. To address this, we propose a framework for learning flatness-preserving residual dynamics in systems whose nominal model admits a pure-feedback form. We show that residuals with a lower-triangular structure preserve both the flatness of the system and the original flat outputs. Moreover, we provide a constructive procedure to recover the flatness diffeomorphism of the augmented system from that of the nominal model.
Building on these insights, we introduce a parameterization of flatness-preserving residuals using smooth function approximators, making them learnable from trajectory data with conventional algorithms.
Our approach is validated in simulation on a 2D quadrotor subject to unmodeled aerodynamic effects. We demonstrate that the resulting learned flat model achieves a tracking error $5\times$ lower than the nominal flat model, while being $20\times$ faster over a structure-agnostic alternative.
\end{abstract}

\section{Introduction}
Differential flatness enables the use of simple and efficient linear techniques for trajectory planning and feedback design in nonlinear systems~\citep{fliess1995flatness}. For systems such as quadrotors, the existence of a flatness diffeomorphism allows the system to be transformed into a chain of integrators, greatly simplifying control~\citep{mellinger2011minimum, greeff2018flatness}. However, flatness is typically established for idealized or nominal models~\citep{welde2023role}, which often fail to capture real-world disturbances and unmodeled dynamics that arise in practice. In certain examples, known perturbations or unmodeled components can be incorporated into an augmented flat model~\citep{faessler2017differential}, but there is currently no general methodology for doing so, or for assessing whether these perturbations will preserve flatness.

Motivated by this shortcoming, we propose augmenting a differentially flat nominal model by learning to approximate the unmodeled (i.e., \emph{residual}) dynamics from trajectory data. However, existing residual learning algorithms~\citep{shi2019neural, carron2019data, torrente2021data, chee2022knode} typically use generic functional forms for the residual, which can lead to augmented dynamics that no longer enjoy the flatness property. Furthermore, even if flatness is preserved, it is often unclear how to construct the flatness diffeomorphism for the resulting system, rendering flatness-based algorithms inapplicable. To overcome this, we restrict our attention to a class of differentially flat nominal dynamics that admit a pure-feedback form and derive a sufficient condition under which the residual preserves flatness. We show that when the residual is lower-triangular---a property that preserves the pure-feedback structure of the nominal dynamics---the flatness diffeomorphisms of the augmented system can be constructed using those of the nominal model. We validate our approach on planar quadrotor dynamics, achieving \rewrite{tracking performance comparable to a nonlinear model predictive controller (NMPC)}{significantly lower tracking error than without learning} while requiring an order of magnitude less computation time \rewrite{}{than a structure-agnostic approach that incorporated the learned residual via nonlinear model predictive control (NMPC)}. To the best of our knowledge, this is the first work to address learning flatness-preserving residual dynamics for differentially flat systems.

\vspace{-.5em}
\subsection{Related Work}
\vspace{-.25em}
\subsubsection{Learning Residual Dynamics}
Leveraging learned residual dynamics in the control loop has been shown to improve tracking performance on various robotic systems, including manipulators~\citep{carron2019data} and quadrotors~\citep{shi2019neural, torrente2021data, chee2022knode}. Prior work incorporated residual dynamics by designing system-specific controllers or applying general techniques such as NMPC. In \citet{shi2019neural}, the authors designed a fix-point iteration-based algorithm tailored to quadrotors to compute the control input. In~\citet{carron2019data, torrente2021data, chee2022knode}, the authors compute the control input with NMPC, which is drastically more computationally expensive than flatness-based control methods for quadrotors~\citep{sun2022comparative}. To address these challenges, we restrict our nominal dynamics to a class of differentially flat pure-feedback systems and seek to learn flatness-preserving residuals, thereby retaining the computational benefits of flatness-based approaches while eliminating the need for case-by-case controller design.

\subsubsection{Learning to Linearize Nonlinear Systems}
Learning-based techniques have been applied to differentially flat systems to derive linearizing coordinate transformations from data.  Representative examples include seeking to identify flatness diffeomorphisms from input-output data~\citep{sferrazza2016numerical,ma2022identification}, and 
directly learning a feedback-linearizing control law~\citep{umlauft2017feedback, westenbroek2019feedback, greeff2020exploiting}. 
For pure-feedback systems, adaptive control methods have also been proposed to deal with parametric uncertainty in the dynamics, including model reference adaptive control~\citep{nam1988model} and neural control methods that leverage backstepping~\citep{ge2002adaptive, ge2004adaptive}. \rewrite{}{In a related vein, model-free control techniques exist for flat systems to estimate and reject more general disturbances online \citep{fliess2013model, join2024flatness}.} In contrast to these approaches, which are complementary to our work, we focus on learning residual dynamics from data while ensuring that the resulting augmented model preserves flatness.

\subsection{Contributions}
Our contributions are three-fold:
\begin{itemize}
    \item \textbf{Flatness-preserving residuals:} We identify a broad class of residual dynamics that preserve the differential flatness of pure-feedback form systems (Theorem~\ref{thm: regularity-perturbed}) and show how to systematically construct the flatness diffeomorphisms for the augmented model from those of the nominal model (Theorem~\ref{thm: form-of-pert-flat-map}).
    \item \textbf{Learning algorithm:} We develop a \rewrite{residual learning framework}{parameterization of the residual dynamics} that enforces this structure, enabling flatness-preserving residual models to be learned directly from trajectory data.
    \item \textbf{Empirical validation:} We demonstrate in simulation that our method \rewrite{matches the tracking performance of NMPC}{significantly reduces tracking error relative to the nominal controller} while \rewrite{achieving an order-of-magnitude improvement in computational efficiency}{being an order-of-magnitude more computationally efficient than an NMPC-based structure-agnostic approach}.
\end{itemize}

\section{Preliminaries and Problem Formulation} \label{sec: problem-formulation}
Consider the nonlinear system
\begin{equation}\label{eq: joint-dynamics}
    \dot\vcx = f(\vcx,\vcu)
\end{equation}
with state $\vcx \in \R^n$ and control input $\vcu \in \R^m$.  We recall the definition of a differentially flat system, as originally introduced by \citet{fliess1995flatness}, but follow the formulation given in \citet{hagenmeyer2003exact}.
%
\begin{definition}[See \citet{hagenmeyer2003exact}]\label{defn: flatness}
    The nonlinear system~\eqref{eq: joint-dynamics} is \textit{(differentially) flat} if a \textit{flat output}
    \begin{equation*}
        \vcy = \Lambda(\vcx, \vcu, \dot{\vcu}, \ddot{\vcu}, \cdots, \vcu^{(\alpha)})
    \end{equation*}
    and a finite number of its time derivatives can be used to reconstruct the state $\vcx$ and control input $\vcu$:
    \begin{align*}
        \vcx &= \Phi(\vcy, \dot{\vcy}, \ddot{\vcy}, \dots, \vcy^{(\beta)}),\\
        \vcu &= \Psi(\vcy, \dot{\vcy}, \ddot{\vcy}, \dots, \vcy^{(\beta)}).
    \end{align*}
    Here, $\Lambda$, $\Phi$, and $\Psi$ are smooth functions, with $(\Phi,\Psi)$ referred to as the \textit{flatness diffeomorphism}. The quantities $\alpha$ and $\beta$ are finite positive integers, and $\vcz^{(c)}$ denotes the $c$-th time derivative of $\vcz$. We say that this system is \textit{locally} flat around a point $(\vcx^*, \vcu^*)$ if the equalities above hold in a neighborhood $\mathcal N$ of $(\vcx^*, \vcu^*)$.
\end{definition}

As outlined in the introduction, differential flatness enables efficient planning and control by allowing for the nonlinear system~\eqref{eq: joint-dynamics} to be rewritten as an equivalent linear one (see, \textit{e.g.},~\citet{fliess1995flatness, murray1995differential,agrawal2021constructive, greeff2018flatness}). However, this procedure often requires exact knowledge of the dynamics $f$, whereas it is often the case that, in practice, only a nominal model $\bar f$ is known.  This motivates the decomposition
\begin{equation}\label{eq: joint-dynamics-perturbed}
    \dot\vcx = f(\vcx,\vcu) = \bar f(\vcx, \vcu) + \Delta(x,u),
\end{equation}
of the dynamics~\eqref{eq: joint-dynamics} into a \textit{nominal dynamics} component, defined by $\bar f(\vcx,\vcu)$, and a \emph{residual dynamics} component
\begin{equation}\label{eq: residual-defn}
    \Delta(\vcx, \vcu) := f(\vcx, \vcu) - \bar{f}(\vcx, \vcu).
\end{equation}
For example, in the context of quadrotor control, the known nominal model might capture the rigid body dynamics, whereas the unknown residual dynamics might capture harder-to-model aerodynamic effects such as rotor drag \citep{kai2017nonlinear, faessler2017differential, shi2019neural, chee2022knode} or another quadrotor's downwash \citep{jain2019modeling, shi2020neural}. 



When these unknown residual dynamics are significant enough to affect control performance, they need to be approximated so that the controller can compensate for their effects.  One common approach is \emph{residual learning}, which leverages data collected from the true system~\eqref{eq: joint-dynamics-perturbed} to learn an estimate $\hat\Delta$ of the residual dynamics~\eqref{eq: residual-defn}, resulting in an \emph{augmented model}
\begin{equation}\label{eq: augment-dynamics-model}
    \hat{f}(\vcx, \vcu) = \bar{f}(\vcx, \vcu) + \hat\Delta(\vcx, \vcu)
\end{equation}
that can be used for improved planning and control.

However, augmenting a differentially flat nominal model $\bar f$ with a learned residual term $\hat \Delta$ will in general not result in a differentially flat augmented model $\hat f = \bar f + \hat \Delta$.
%
%
We therefore seek to identify general parameterizations of residual dynamics $\hat \Delta$ such that the augmented nominal model $\hat f$
is differentially flat if the nominal model $\bar f$ is differentially flat.
Towards that end, we restrict our attention to nominal dynamics that can be expressed in pure-feedback form.

\begin{assumption}[Pure-Feedback Form Nominal Dynamics]\label{assm: nominal-pure-feedback}
    The nominal dynamics $\dot\vcx = \bar f(\vcx,\vcu)$ can be put into a multi-input multi-output nonlinear \textit{pure-feedback} form
    \begin{equation}\label{eq: flat-PFF}
        \begin{aligned}
            \dot \vcx_1 &= \bar{f}_1(\vcx_1, \vcx_2), \\
            \dot \vcx_2 &= \bar{f}_2(\vcx_1, \vcx_2, \vcx_3), \\
            \vdots\\
            \dot \vcx_r &= \bar{f}_r(\vcx_1, \ldots, \vcx_r, \vcu),
        \end{aligned}
    \end{equation}
    where $\vcx_i \in \R^m$ are (vector) sub-states that decompose the joint state $\vcx$ as ${\vcx:=[\vcx_1^\top, \ldots, \vcx_r^\top]^\top}$ and $\vcu\in\R^m$ are the control inputs. The sub-state nominal dynamics $\bar{f}_i$ are respectively of dimensions $\bar{f}_i: \R^{(i+1)m} \to \R^m$. 
\end{assumption}
\begin{remark}\label{rem: uniform-rel-degree}
    Assumption~\ref{assm: nominal-pure-feedback} implicitly assumes that the $m$ sub-states $\vcx_{i, 1}, \ldots, \vcx_{i,m}$ have the same relative degree $r$, and hence all $\vcx_i$ share the same dimension with $\vcu$. We choose this formulation for ease of notation and exposition, but our analysis can be adapted to the case where the sub-states have different relative degrees via dynamic extension \rewrite{(See the running example \eqref{eq: 2dquad-original})}{}.
\end{remark}

Given the above, we next impose a regularity condition which ensures that pure-feedback form nominal dynamics are flat. We then define a parameterization of residual dynamics models ensuring that the resulting augmented model is flat if the nominal model is flat. We also show how to construct flatness diffeomorphisms for the augmented model from those
of the nominal model. Finally, we present our flatness-preserving residual learning algorithm and illustrate it via numerical experiments on a 2D quadrotor system.

\section{Differentially Flat Pure-Feedback Systems} \label{sec: PFF-systems}
We recall the well-known fact that pure-feedback systems \eqref{eq: flat-PFF} are locally differentially flat under suitable regularity conditions \citep{murray1995differential}, and provide a proof sketch (see Appendix~\ref{sec: thm1-proof} for a full proof), as the details will later prove useful in our extension to residual learning. 
\begin{assumption}[Regularity of Nominal Dynamics]\label{assm: nominal-dynamics-regular}
    There exists a neighborhood $\mathcal{N}$ of the desired operating point $(\vcx^*, \vcu^*)$, on which the nominal dynamics \eqref{eq: flat-PFF} satisfy that
    \begin{itemize}
        \item $\bar{f}_i$ are continuously differentiable, and
        \item The following partial Jacobians\footnote{We use $D_\vcz h(\cdot)$ to denote the partial Jacobian matrix of a function $h$ with respect to $\vcz$. In what follows, we also use $Dh(\cdot)$ to denote the full Jacobian of $h$ with respect to all its arguments.} are non-singular.
        \begin{align*}
            \left\lvert D_{\vcx_{i+1}}\bar{f}_i(\vcx^*_1, \ldots,\vcx^*_{i+1}) \right\rvert &\neq 0, \quad i=1, \ldots, r-1,\\
            \left\lvert D_{\vcu}\bar{f}_r(\vcx^*_1, \ldots,\vcx^*_{r}, \vcu^*) \right\rvert &\neq 0.
        \end{align*}
    \end{itemize}
\end{assumption}
\begin{theorem}\label{thm: regularity-unperturbed}
Let the nominal model $\bar f$ 
satisfy Assumptions~\ref{assm: nominal-pure-feedback} and \ref{assm: nominal-dynamics-regular}.
    Then $\bar f$ is locally differentially flat around $(\vcx^*, \vcu^*)$
    with flat output $\vcy = \vcx_1$.
    %
\end{theorem}
\begin{proof}[Proof Sketch]
By the implicit function theorem, the regularity of $\bar{f}_i$ guarantees that one can locally solve for $\vcx_{i+1}$ in terms of $\vcx_1, \ldots, \vcx_i$ and $\dot\vcx_i$. That is, locally, there exists a unique, continuously differentiable function $h_i$ such that
\begin{equation*}
    h_i(\vcx_1, \ldots, \vcx_i, \dot\vcx_i) = \vcx_{i+1}.
\end{equation*}
Thus, starting from $\vcy = \vcx_1$, one can locally solve for the state and control variables recursively. For example, the state $\vcx_2$ is given by
\begin{equation*}
    \vcx_2 = h_1(\vcx_1, \dot\vcx_1) = h_1(\vcy, \dot\vcy).
\end{equation*}
Noting that, by the chain rule, $\dot\vcx_2$ satisfies
\begin{equation*}
    \dot\vcx_2 = \frac{d}{dt}h_1(\vcy, \dot\vcy) = Dh_1(\vcy, \dot\vcy)
    \begin{bmatrix}
        \dot\vcy \\ \ddot\vcy
    \end{bmatrix},
\end{equation*}
one can solve for $\vcx_3$ as a function of $\vcy, \dot\vcy$ and $\ddot\vcy$. Recursively applying this procedure to $\vcx_1,\dots,\vcx_r$, one can construct the flatness diffeomorphism $(\Phi, \Psi)$. 
\end{proof}
Many well-known flat systems (\textit{e.g.}, unicycles and fully-actuated manipulators\cite{rathinam1998configuration}) can be placed into a form that complies with Assumptions~\ref{assm: nominal-pure-feedback} and \ref{assm: nominal-dynamics-regular} under suitable dynamic extensions.
Furthermore, while Theorem~\ref{thm: regularity-unperturbed} characterizes local differential flatness around an operating point, many well-known flat systems are differentially flat over a large subset of the state and control spaces. 
%
%
\ifniceformat
\begin{wrapfigure}{r}{0.42\textwidth}
    \centering
    \includegraphics[width=0.9\linewidth]{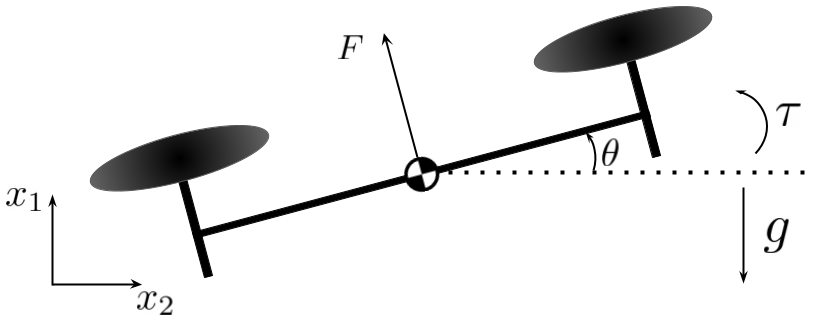}
    \caption{Diagram of the 2D Quadrotor System.}
    \label{fig: 2d-quad-graphic}
\end{wrapfigure}
\else
\begin{figure}[t]
    \centering
    \includegraphics[width=0.9\linewidth]{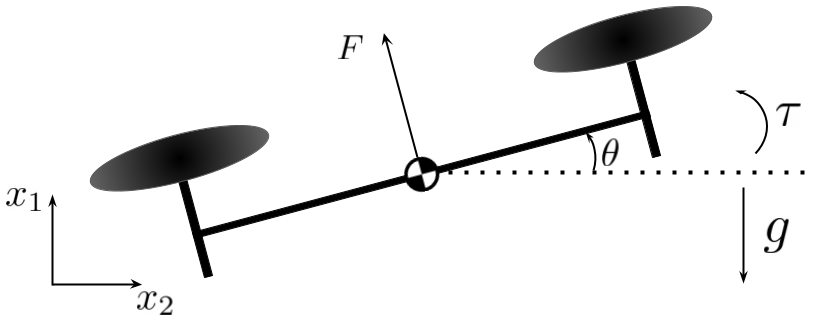}
    \caption{Diagram of the 2D Quadrotor System.}
    \label{fig: 2d-quad-graphic}
\end{figure}
\fi
\begin{example}[Running Example: 2D Quadrotor]
    Consider a quadrotor constrained to move only within a vertical plane (as shown in Fig.~\ref{fig: 2d-quad-graphic}). A common nominal model for this system's dynamics (derived from first principles) is given by
    \begin{equation}\label{eq: 2dquad-original}
        \dot \vcp = \vcv, \quad \mathfrak{m}\dot \vcv = \begin{bmatrix} -F\sin\theta \\ F\cos\theta\end{bmatrix} - \vcg, \quad
        \dot \theta = \omega, \quad  I\dot \omega = \tau.
    \end{equation}
    Here $\vcp, \vcv \in \R^2$ are the position and linear velocity of the quadrotor. $\theta \in [-\pi, \pi)$ is its orientation, and $\omega \in \R$ its angular velocity. Its mass and moment of inertia are given as positive scalar constants $\mathfrak m$ and $I$. The control inputs are the combined thrust $F \in \R$ and torque $\tau \in \R$, and $\vcg = [0\;g]^\top$ denotes gravitational acceleration. This system can be written in the pure-feedback form \eqref{eq: flat-PFF} by dynamically extending the thrust $F$ twice, resulting in the states and control action
    \begin{equation*}
        \vcx_1 = \vcp,\; \vcx_2=\vcv,\; \vcx_3 = [F, \theta]^\top,\; \vcx_4 = [\dot F, \omega]^\top, \vcu = [\ddot F, \tau]
    \end{equation*}
    and the dynamics%
    \begin{alignat}{3}
            &\dot \vcx_1 = \bar{f}_1(\vcx_2) = \vcx_2, \ \  &&\dot\vcx_2 = \bar{f}_2(\vcx_3) &&= 
            \frac{1}{\mathfrak m}\begin{bmatrix} -F\sin\theta 
            \\ F\cos\theta - \mathfrak{m}g \end{bmatrix}, \nonumber \\
            &\dot \vcx_3 = \bar{f}_3(\vcx_4) = \vcx_4, \ \  &&\dot \vcx_4 = \bar{f}_4(\vcu) &&= \begin{bmatrix} 1 & 0 \\ 0 & I^{-1} \end{bmatrix}\vcu.\qquad
            \label{eq: 2dquad}
    \end{alignat}
    By Theorem~\ref{thm: regularity-unperturbed}, it is clear that the system is locally differentially flat with flat output $\vcy = \vcx_1$ except for when $F=0$.
\end{example}
\section{Flatness-Preserving Perturbations}\label{sec: flat-preserve-res}
%
We show that flat nominal dynamics satisfying the assumptions of Theorem~\ref{thm: regularity-unperturbed} remain differentially flat under a class of additive residual dynamics that preserve the pure-feedback structure. Further, this class of residual dynamics allows us to construct the flatness diffeomorphisms of the augmented system from those of the nominal system.

We first define a class of \emph{lower-triangular} residual dynamics that are compatible with the pure-feedback form of the nominal dynamics~\eqref{eq: flat-PFF}.
\begin{definition}[Lower-Triangular Residual Dynamics]\label{defn: residual-strict-feedback}
      Residual dynamics $\Delta(\vcx)$ that can be expressed as
    \begin{equation}\label{eq: perturbed-PFF}
        \Delta(\vcx) = 
        \begin{bmatrix}
            \Delta_1(\vcx_1) \\
            \Delta_2(\vcx_1, \vcx_2) \\
            \vdots\\
            \Delta_r(\vcx_1, \ldots, \vcx_r)
        \end{bmatrix},
    \end{equation}
    for continuously differentiable $\Delta_i$, $i=1,\dots,r$, are called \emph{lower-triangular}.
\end{definition}
%
While this structure may seem overly restrictive, particular instances satisfying these assumptions have already been identified in the literature. We highlight one such example in the context of unmodeled aerodynamics for aerial vehicles at the end of this section.

The next theorem shows that restricting ourselves to lower-triangular residual dynamics, which preserve the pure-feedback structure and regularity of the nominal dynamics~\eqref{eq: flat-PFF}, ensures the flatness of the resulting augmented model.
%
\begin{theorem}\label{thm: regularity-perturbed}
    Let the nominal model $\bar f$ satisfy Assumptions~\ref{assm: nominal-pure-feedback} and \ref{assm: nominal-dynamics-regular}
    and let $\hat\Delta$ be lower-triangular.  Then 
    the augmented dynamics $\bar f + \hat\Delta$ are locally differentially flat around $(\vcx^*, \vcu^*)$ with flat output $\vcy = \vcx_1$.
\end{theorem}
\begin{proof}
Due to the lower-triangular structure of the residual, the perturbed dynamics $\bar f + \hat\Delta$ are still in pure-feedback form, and share relevant Jacobians with the nominal dynamics:
\ifniceformat
\begin{equation*}
    D_{\vcx_{i+1}}\Big( \bar{f}_i(\vcx_1, \ldots,\vcx_{i+1}) + \hat\Delta_i(\vcx_1, \ldots,\vcx_i) \Big) \\=  D_{\vcx_{i+1}} \bar{f}_i(\vcx_1, \ldots,\vcx_{i+1}),
\end{equation*}
\else
\begin{multline*}
    D_{\vcx_{i+1}}\Big( \bar{f}_i(\vcx_1, \ldots,\vcx_{i+1}) + \hat\Delta_i(\vcx_1, \ldots,\vcx_i) \Big) \\=  D_{\vcx_{i+1}} \bar{f}_i(\vcx_1, \ldots,\vcx_{i+1}),
\end{multline*}
\fi
as $\hat\Delta_i$ does not depend on $\vcx_{i+1}$. Thus local differential flatness of the perturbed dynamics $\bar f + \hat\Delta$ follows directly from Theorem~\ref{thm: regularity-unperturbed}.
\end{proof}

Note that this class of residual dynamics also preserves the flat outputs of the nominal system. This is desirable, as the flat outputs often carry semantic meaning that benefits downstream planning and control tasks.

We now show how to construct flatness diffeomorphisms for the augmented system given access to the flatness diffeomorphisms of the nominal system.
\begin{assumption}\label{assm: nominal-flat-map}
    We assume we know the flatness diffeomorphisms of the nominal system \eqref{eq: flat-PFF}. In particular, we assume that the set of continuously differentiable functions ${h_k: \R^{(k+1)m}\to\R^m}$ that satisfy
    \ifniceformat
    \begin{equation}\label{eq: defn-of-g}
        \begin{aligned}
            &h_k\left(\vcx_1, \ldots, \vcx_{k}, \bar{f}_k(\vcx_1, \ldots, \vcx_{k+1})\right) = \vcx_{k+1}, \qquad k = 1, \ldots, r-1,\\
            &h_{r}\left(\vcx_1, \ldots, \vcx_{k}, \bar{f}_r(\vcx_1, \ldots, \vcx_{r}, \vcu)\right) = \vcu
        \end{aligned}
    \end{equation}
    \else
    \begin{equation}\label{eq: defn-of-g}
        \begin{aligned}
            &h_k\left(\vcx_1, \ldots, \vcx_{k}, \bar{f}_k(\vcx_1, \ldots, \vcx_{k+1})\right) = \vcx_{k+1}, \\
            &\qquad\qquad\qquad\qquad\qquad\qquad k = 1, \ldots, r-1,\\
            &h_{r}\left(\vcx_1, \ldots, \vcx_{k}, \bar{f}_r(\vcx_1, \ldots, \vcx_{r}, \vcu)\right) = \vcu
        \end{aligned}
    \end{equation}
    \fi
    for all ${(\vcx, \vcu) \in \mathcal{N}}$ for the nominal system~\eqref{eq: flat-PFF} is known.
\end{assumption}
As seen in Theorem~\ref{thm: regularity-unperturbed}, the implicit function theorem guarantees the existence of the maps $h_i$ under Assumption~\ref{assm: nominal-dynamics-regular}, but does not provide a way to construct them.  Fortunately, such functions have already been identified for many flat systems of interest (see, \textit{e.g.},~\citet{murray1995differential}).
\begin{example}[Running Example: 2D Quadrotor]
    For the pure-feedback 2D quadrotor \eqref{eq: 2dquad}, the desired maps described in Assumption~\ref{assm: nominal-flat-map} are given by
\begin{equation*}
    \begin{aligned}
    &h_1(\vcx_1, \dot\vcx_1) = \dot\vcx_1, \\
    &h_2(\vcx_1, \vcx_2, \dot\vcx_2) =
    \begin{bmatrix}
        \mathfrak m \sqrt{(\dot\vcx_{2, 1})^2 + (\dot\vcx_{2,2} + g)^2} \\
        \arctan2(-\dot\vcx_{2, 1},\dot\vcx_{2,2} + g)
    \end{bmatrix},\\
    &h_3(\vcx_1, \vcx_2, \vcx_3, \dot\vcx_3) = \dot\vcx_3,\\
    &h_4(\vcx_1, \vcx_2, \vcx_3, \vcx_4, \dot\vcx_4) = \begin{bmatrix} 1 & 0 \\ 0 & I \end{bmatrix} \dot\vcx_4,
    \end{aligned}
\end{equation*}
which are valid over the set $\{(\vcx, \vcu) \in \R^{10}:\vcx_{3, 1} = F > 0, \vcx_{3, 2} \in [-\pi, \pi)\}$.
\end{example}
Similar to the proof of Theorem~\ref{thm: regularity-unperturbed}, we construct the flatness diffeomorphisms for the augmented system recursively.
\begin{theorem}\label{thm: form-of-pert-flat-map}
    Let the assumptions of Theorem~\ref{thm: regularity-perturbed} hold. Define the continuously differentiable functions $\hat\Phi_i,$ in terms of the functions $h_i$ given in equation~\eqref{eq: defn-of-g} as follows:
    \ifniceformat
    \begin{equation}\label{eq: pert-flat-map-construction}
    \begin{aligned}
        \hat\Phi_1(\vcy) &= \vcy, \\
        \hat\Phi_k(\vcy,\ldots,\vcy^{(k-1)}) &= h_{k-1}\left(\hat\Phi_1, \ldots, \hat\Phi_{k-1},
        D\hat\Phi_{k-1}
        \begin{bmatrix}
            \dot\vcy \\ \vdots \\ \vcy^{(k-1)}.
        \end{bmatrix} - \hat\Delta_{k-1}(\hat\Phi_1, \ldots, \hat\Phi_{k-1}) \right),
    \end{aligned}
    \end{equation}
    \else
    \begin{equation}\label{eq: pert-flat-map-construction}
    \begin{aligned}
        &\hat\Phi_k(\vcy,\ldots,\vcy^{(k-1)}) = h_{k-1}\Big(\hat\Phi_1, \ldots, \hat\Phi_{k-1},\\
        &\qquad\qquad D\hat\Phi_{k-1}
        \begin{bmatrix}
            \dot\vcy \\ \vdots \\ \vcy^{(k-1)}.
        \end{bmatrix} - \hat\Delta_{k-1}(\hat\Phi_1, \ldots, \hat\Phi_{k-1}) \Big),
    \end{aligned}
    \end{equation}
    \fi%
    for $k=2, \ldots, r+1$, wherein we omit the arguments of the $\hat\Phi_k$ on the right-hand side to reduce notational clutter.
    Then, under the augmented dynamics $\bar f + \hat \Delta$, we have
    \begin{equation*}
        \begin{aligned}
            &\vcx_k = \hat\Phi_k(\vcy, \ldots, \vcy^{(k-1)}),\quad k=1,\ldots,r\\
            &\vcu = \hat\Phi_{r+1}(\vcy, \ldots, \vcy^{(r)})
        \end{aligned}
    \end{equation*}
    for $(\vcx, \vcu) \in \mathcal{N}$.
\end{theorem}
\begin{proof}
    We prove the claim by induction on $k$. For the base case $k=1$, we have from Theorem~\ref{thm: regularity-perturbed} that $\vcy = \vcx_1$, and thus ${\vcx_1 = \hat\Phi_1(\vcy) = \vcy}$. For the induction step, assume that for all $k' < k \leq r$, we have
    \begin{equation*}
        \vcx_{k'} = \hat\Phi_{k'}(\vcy, \ldots, \vcy^{(k'-1)}).
    \end{equation*}
    It suffices to show that the equality also holds for ${\hat\Phi_k}$. Note that \textit{under the augmented dynamics}, we have
    \begin{equation*}
        \bar{f}_{k-1}(\vcx_1, \ldots, \vcx_{k}) = \dot\vcx_{k-1} - \hat\Delta_{k-1}(\vcx_1, \ldots, \vcx_{k-1}).
    \end{equation*}
    Thus, by the definition of $h_i$ \eqref{eq: defn-of-g},
    \begin{equation*}
        h_{k-1}\left(\vcx_1, \ldots, \vcx_{k-1}, \dot\vcx_{k-1} - \hat\Delta_{k-1}(\vcx_1, \ldots, \vcx_{k-1})\right) = \vcx_k.
    \end{equation*}
    From the induction hypothesis, we recover the expression in \eqref{eq: pert-flat-map-construction}, and thus
    \begin{equation*}
        \vcx_{k} = \hat\Phi_{k}(\vcy, \ldots, \vcy^{(k-1)}).
    \end{equation*}
    Since $\hat\Phi_{k'}, \Delta_{k-1}$, and $h_{k-1}$ are continuously differentiable, ${\hat{\Phi}_k}$ is also continuously differentiable. With a slight abuse of notation, the desired result for $\hat\Psi$ follows from the same argument if we denote $\vcx_{r+1}:= \vcu$.
\end{proof}
\begin{example}[Running Example: 2D Quadrotor]
    Consider the lower-triangular perturbation 
    \begin{equation*}
        \Delta_1 = \Delta_3 = \Delta_4 = 0,\quad
        \Delta_2(\vcx_2) = -\frac{C_r}{\mathfrak{m}}\vcv = -\frac{C_r}{\mathfrak{m}}\vcx_2
    \end{equation*}
    to the \emph{nominal} 2D quadrotor system~\eqref{eq: 2dquad}, capturing an isotropic linear drag similar to rotor drag \citep{faessler2017differential}. The resulting flat map can be constructed per \eqref{eq: pert-flat-map-construction} as
    \begin{align*}
        \hat\Phi_1(\vcy) &= \vcy,\\
        \hat\Phi_2(\vcy, \dot\vcy)
        &= h_1\left(\hat\Phi_1(\vcy), D\hat\Phi_1(\vcy)\;\dot\vcy - \Delta_1(\vcy)\right) = \dot\vcy,\\
        \hat\Phi_3(\vcy, \dot\vcy, \ddot\vcy) &= h_2\left(\hat\Phi_1, \hat\Phi_2, D\hat\Phi_2\;[\dot\vcy\;\;\ddot\vcy]^\top - \Delta_2\right)\\
        &= \begin{bmatrix}
        \mathfrak{m}\sqrt{(\ddot\vcy_{1} + \frac{C_r}{\mathfrak{m}}\dot\vcy_1)^2 + (\ddot\vcy_{2} + g + \frac{C_r}{\mathfrak{m}}\dot\vcy_2)^2} \\
        \rewrite{}{\tan^{-1}\left(\left( \ddot\vcy_{2} + g + \frac{C_r}{\mathfrak{m}}\dot\vcy_2 \right) / \left(-\ddot\vcy_{1} + \frac{C_r}{\mathfrak{m}}\dot\vcy_1\right) \right)}
    \end{bmatrix}.
    \end{align*}
    $\hat\Phi_4$ and $\hat\Psi$ can then be found by taking time derivative of $\hat\Phi_3$, which we omit due to space constraints.
\end{example}

\section{Learning Flatness-Preserving Residuals} \label{sec: learning}
The results in Section~\ref{sec: flat-preserve-res} identify a rich class of residual dynamics (Def.~\ref{defn: residual-strict-feedback}) that can be added to a flat nominal system in pure-feedback form so as to define augmented systems that are also differentially flat (Thm.~\ref{thm: regularity-perturbed}), and for which the corresponding flatness diffeomorphisms can be easily constructed given those of the nominal system (Thm.~\ref{thm: form-of-pert-flat-map}).  We leverage these insights to propose an algorithm for learning flatness-preserving residuals from data.

In particular, we assume access to a set of trajectories ${\vctau = \{\vctau^i\}_{i=1}^{N}}$. Each trajectory $\vctau^i$ consists of a sequence of state and control input pairs
\begin{equation*}
\vctau^i = \Big( \big(\vcx^i(t_1), \vcu^i(t_1)\big), , \ldots, \big(\vcx^i(t_T), \vcu^i(t_T)\big)\Big)
\end{equation*}
sampled at times $\{t_1, \ldots, t_T\}$ generated by the \emph{true dynamics}~\eqref{eq: joint-dynamics-perturbed}. We aim to learn a flatness-preserving residual dynamics $\hat\Delta_\Theta(\vcx, \vcu)$, defined by learnable parameters $\Theta\in\R^p$, that best augment the nominal model to match the trajectories $\vctau$ generated by the true system.

So as to leverage Theorems~\ref{thm: regularity-perturbed} and \ref{thm: form-of-pert-flat-map}, we restrict the learnable residual dynamics model $\hat{\Delta}_\Theta$ to be lower-triangular:
\begin{equation*}
    \hat{\Delta}_\Theta(\vcx) = \begin{bmatrix}
        \hat\Delta_1(\vcx_1; \theta_1)\\
        \hat\Delta_2(\vcx_1, \vcx_2; \theta_2) \\
        \vdots\\
        \hat\Delta_r(\vcx_1, \ldots, \vcx_r; \theta_r)
    \end{bmatrix},
\end{equation*}
where $\Theta := (\theta_1, \ldots, \theta_r)$ are parameters of suitable dimensions
and the $\hat\Delta_i$ are continuously differentiable in each $\vcx_j$. By Theorem~\ref{thm: regularity-perturbed}, the augmented system $\bar{f}+\hat{\Delta}_\Theta$ is differentially flat, and its flatness diffeomorphisms can be constructed according to equation~\eqref{eq: pert-flat-map-construction}.
{A practically appealing and standard choice is to parameterize the $\hat\Delta_i$ as neural networks with continuously differentiable activation functions.}

To find the $\hat\Delta_\Theta$ that best approximates the data, we consider the loss function that sums the squared error between the true and learned state derivatives along each trajectory:
\begin{equation}\label{eq: loss-fn}
    L(\Theta) = \sum_{i=1}^{N}\sum_{k=1}^{T-1}
    \norm{\Delta(\vctau^i_k) -
    \hat\Delta(\vctau^i_k; \Theta) }_2^2,
\end{equation}
where the true residual $\Delta$ is approximated via a one-step finite difference as
\begin{equation*}
    \Delta(\vctau^i_k) := \frac{\vcx^i(t_{k+1}) - \vcx^i(t_k) }{t_{k+1} - t_k} - \bar{f}(\vcx^i(t_k), \vcu^i(t_k)).
\end{equation*}
The weights $\Theta$ can then be optimized via backpropagation using readily available deep-learning libraries. After the disturbance is learned, the flatness diffeomorphism can be constructed offline using the formula described in \eqref{eq: pert-flat-map-construction}, where we again leverage the deep learning libraries to compute the relevant Jacobians via automatic differentiation.






\section{Numerical Experiments}\label{sec: experiment}
We validate our algorithm via simulation experiments on the 2D quadrotor \eqref{eq: 2dquad} under unmodeled air resistance. First, we test whether augmenting the nominal model with learned residual dynamics can generate more dynamically consistent trajectories in an open-loop rollout (Section~\ref{sec: exp-open-loop}). Then, we design a flatness-based tracking controller using the learned model and compare its tracking performance and computation time against an NMPC baseline (Section~\ref{sec: exp-closed-loop}). We show that the flatness-based controller achieves tracking performance similar to that of the NMPC, while requiring an order of magnitude less time to compute the control action. This computational efficiency will allow it to be run at a higher frequency and makes it more suitable for computationally constrained hardware. Code to reproduce our experiments is available at \url{https://github.com/FJYang96/flat-residual/}.

\subsection{Setup}
We consider the 2D quadrotor dynamics \eqref{eq: 2dquad} under the lower-triangular perturbation
\begin{equation}\label{eq: 2d-quad-disturbance}
\begin{aligned}
    \dot\vcx_2 &= \bar{f}_2(\vcx_3) + \Delta_2(\vcx_2) \\
    &= \bar{f}_2(\vcx_3) - \underbrace{C_r\vcx_2}_{\Delta^r(\vcx_2)}
    - \underbrace{C_p\norm{\vcx_2}\vcx_2}_{\Delta^p(\vcx_2)},
\end{aligned}
\end{equation}
where $\Delta^p$ and $\Delta^r$ seek to capture the respective effects of a (fuselage) parasitic drag and an isotropic linear drag inspired by rotor drag. These terms are simplified from their 3D, non-isotropic counterparts considered in previous work~\citep{kai2017nonlinear, faessler2017differential, chee2022knode}. The values of mass, moment of inertia, and drag coefficients are detailed in Appendix~\ref{sec: exp-detail}.

We collect $3000$ trajectories by sampling the initial condition in the square region $\vcp \in [-1, 1]^2$ and applying random control inputs. Details on the sampling procedure for both the initial condition and the control inputs can be found in Appendix~\ref{sec: exp-detail}. Each trajectory is $0.5$ seconds and sampled at $100Hz$. 
We parameterize $\hat\Delta_\Theta(x)$ as a feedforward neural network (1 hidden layer of 32 neurons) with GeLU activations to ensure smoothness. Training for 20 epochs with ADAM on an Nvidia A5000 GPU takes under 1 minute.
To reduce the effect of randomness in the trajectory data samples and network initialization, we run the following experiments for $30$ different seeds.

In what follows, we couple the learned model with common flatness-based techniques for planning and control. We evaluate the performance on a circular trajectory of the form
\begin{equation*}
    \vcp_{circ}(t) = \begin{bmatrix}
        \cos (\omega t) \\ \sin (\omega t)
    \end{bmatrix}, \quad t = [0, T]
\end{equation*}
and a Genoro's lemniscate trajectory of the form
\begin{equation*}
    \vcp_{lem}(t) = \begin{bmatrix}
        \frac{a\cos(\omega t)}{1 + \sin^2(\omega t)} \\ \frac{b\sin(\omega t)\cos(\omega t)}{1 + \sin^2(\omega t)}
    \end{bmatrix}, \quad t = [0, T],
\end{equation*}
with $a=1, b=0.6$. We take $T = 14$ and set $\omega = 2\pi/T$.

\subsection{Open-Loop Trajectory Evaluation}\label{sec: exp-open-loop}
We first test the augmented model's ability to generate dynamically consistent \textit{state and input} trajectories from trajectories of the \textit{flat outputs}. For each reference trajectory ($\vcp_{circ}$ or $\vcp_{lem}$), we generate open-loop control inputs in two ways: (1) using the flatness diffeomorphism, $(\bar\Phi, \bar\Psi)$, of the nominal 2D quadrotor dynamics model $\bar f$ (which ignores drag), and (2) using the augmented diffeomorphism, $(\hat\Phi, \hat\Psi)$, of the augmented dynamics model $\bar f + \hat\Delta_\Theta$ that accounts for drag. We then roll out the open loop control actions under the true dynamics $f$. Recall that for our system \eqref{eq: 2dquad}, the flat outputs are the position (\textit{i.e.}, ${\vcy = \vcx_1 = \vcp}$). For the nominal diffeomorphism, we compute the trajectory $\bar{\vcx}(t)$ that satisfies
\begin{equation*}
    \frac{d}{dt}\bar{\vcx}(t) = f\left(
    \bar{\vcx}(t), \bar{\Psi}\left(\vcp(t), \ldots, \vcp^{(4)}(t)\right)
    \right).
\end{equation*}
Similarly, we compute the trajectory $\hat{\vcx}(t)$ that satisfies
\begin{equation*}
    \frac{d}{dt}\hat{\vcx}(t) = f\left(
    \hat{\vcx}(t), \hat{\Psi}\left(\vcp(t), \ldots, \vcp^{(4)}(t)\right)
    \right)
\end{equation*}
for the augmented system.
%
%
In both cases, we integrate these ODE's numerically using the Runge-Kutta method with a step size $\tau=10^{-2}$ seconds. We report the position error, computed for a trajectory $\vcx(t)$ as
\begin{equation*}
    \vce = \frac{1}{\lceil T/\tau \rceil}\sum_{k=1}^{\lceil T/\tau \rceil} \norm{\vcx_1(k\tau) - \vcp(k\tau)}
\end{equation*}
in Table~\ref{tab:open-loop-result}, and show representative trajectories in Figure~\ref{fig: open-loop-plot}. We note that, on average, incorporating the learned residual model reduces position error by nearly an order of magnitude. This suggests that the learned flat model is able to capture the unmodeled effects and adjust the flatness diffeomorphisms accordingly, showing its potential to be used for accurate trajectory planning. We also repeated the same steps for a flatness diffeomorphism $(\Phi^*, \Psi^*)$ constructed with the ground-truth residual. Consistent with our theoretical results, the corresponding simulated open-loop trajectory has a position error less than $10^{-5}$, providing further evidence that the procedure in \eqref{eq: pert-flat-map-construction} correctly constructs the diffeomorphism and thus corroborating our claims in Theorem~\ref{thm: form-of-pert-flat-map}.

\ifniceformat
\begin{figure}[ht]
    \centering
    \begin{subfigure}{0.22\textwidth}
        \centering
        \includegraphics[width=\linewidth, trim=10pt 10pt 10pt 10pt, clip]{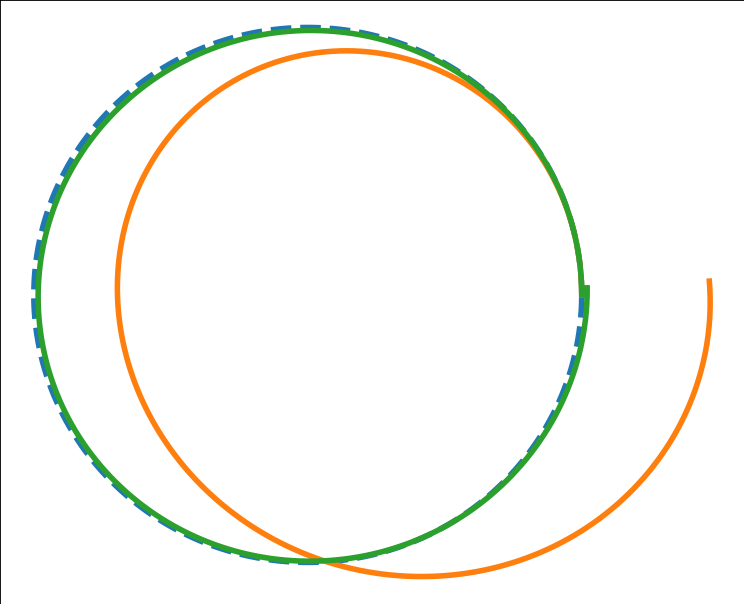}
        \caption{Circular Trajectory}
    \end{subfigure}
    \begin{subfigure}{0.25\textwidth}
        \centering
        \includegraphics[width=\linewidth, trim=5 70 10 5, clip]{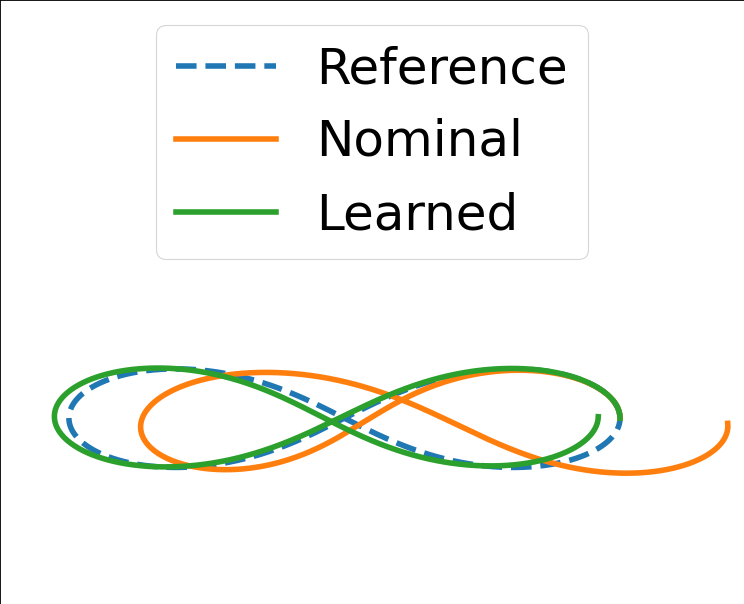}
        \caption{Lemniscate Trajectory}
    \end{subfigure}
    \caption{Representative Open-Loop Trajectories: Adjusting the flatness diffeomorphism to account for the learned residual dynamics results in more accurate trajectory planning, leading to open-loop behavior that better matches the intended reference trajectory.}
    \label{fig: open-loop-plot}
\end{figure}

\else

\begin{figure}[ht]
    \centering
    \begin{subfigure}{0.22\textwidth}
        \centering
        \includegraphics[width=\linewidth, trim=10pt 10pt 10pt 10pt, clip]{img/ol-ellipse.png}
        \caption{Circular Trajectory}
    \end{subfigure}
    \hfill
    \begin{subfigure}{0.25\textwidth}
        \centering
        \includegraphics[width=\linewidth, trim=5 70 10 5, clip]{img/ol-lem.png}
        \caption{Lemniscate Trajectory}
    \end{subfigure}
    \caption{Representative Open-Loop Trajectories: Adjusting the flatness diffeomorphism to account for the learned residual dynamics results in more accurate trajectory planning, leading to open-loop behavior that better matches the intended reference trajectory.}
    \label{fig: open-loop-plot}
\end{figure}
\fi

\begin{table}[H]
    \centering
    \begin{tabular}{|c|c|c|c|c|}
    \hline            &   \textbf{Nominal}  &  \textbf{Learned} \\ \hline
    \textbf{Circle}     &  0.312  &  0.0265 $\pm$ 0.0058 \\ \hline
    \textbf{Lemniscate}  &  0.238  &  0.0417 $\pm$ 0.0066 \\ \hline
    \end{tabular}
    \caption{Mean Open-Loop Position Error (m) $\pm$ Standard Deviation over $30$ training runs.}
    \label{tab:open-loop-result}
\end{table}

\subsection{Closed-Loop Tracking Control}\label{sec: exp-closed-loop}
We now use the learned model to construct a tracking controller for the 2D quadrotor that leverages the flatness of the learned augmented dynamics.  As is standard, we design a linear feedback control law for the linear flat dynamics (sometimes called the ``Brunovsky canonical form''~\citep{fliess1995flatness})
\begin{equation*}
    \vcy^{(4)} = \vcnu,
\end{equation*}
where $\vcnu$ should be interpreted as a virtual control input.  For a reference trajectory $\vcy_r$, we choose the virtual control input $\vcnu$ according to the control law
\begin{equation}\label{eq: flat-controller-gain}
    \vcnu = \vcy_r^{(4)} -k_0(\vcy - \vcy_r)- \ldots -
            k_{3}(\vcy^{(3)} - \vcy_r^{(3)}).
\end{equation}
We assume that the controller can observe the state $\vcx$ and use a Luenberger observer to estimate the higher-order derivatives of the flat output. We then compute the input as
\begin{equation*}
    \vcu = \hat\Psi\left( \vcy, \dot\vcy, \ddot \vcy, \vcy^{(3)}, \vcnu \right).
\end{equation*}

We compare this flatness-based approach against a nonlinear model predictive controller (NMPC) that leverages the same learned model. Given a measured state $\xi$, the NMPC controller solves the nonlinear optimal control problem
\begin{equation*}
\begin{aligned}
    \minimize{\vcx, \vcu} &\quad \int_{t=0}^{H}\Big((\vcy-\vcy_r)^\top Q (\vcy-\vcy_r) + \vcu^\top R \vcu \Big)dt \\
    &\qquad\quad + (\vcy(H) - \vcy_r(H))^\top Q_T (\vcy(H)-\vcy_r(H)) \\
    \text{subject to} &\quad \dot\vcx = \bar{f}(\vcx, \vcu) + \hat\vcDelta_\Theta(\vcx)\\
    &\quad \vcy = \vcx_1, \qquad \vcx(0) = \xi.
\end{aligned}
\end{equation*}
The controller then applies the control action $\vcu(0)$ until the next sampling step. In practice, we solve this problem using the direct method. We consider a prediction horizon of ${H=1}$ second and discretize the dynamics into $100$ time steps using RK4 integration. We use the cost matrices $Q = Q_T = I$ and $R = 0.001I$. We warmstart the optimization problem from the previous solution and solve it in Python using Casadi~\citep{Andersson2019} until convergence. We also assume that the controller observes the full state at each time step.

\begin{figure*}[ht]
    \centering
    \begin{subfigure}{0.25\textwidth}
        \centering
        \includegraphics[width=\linewidth, trim=10 10 10 10, clip]{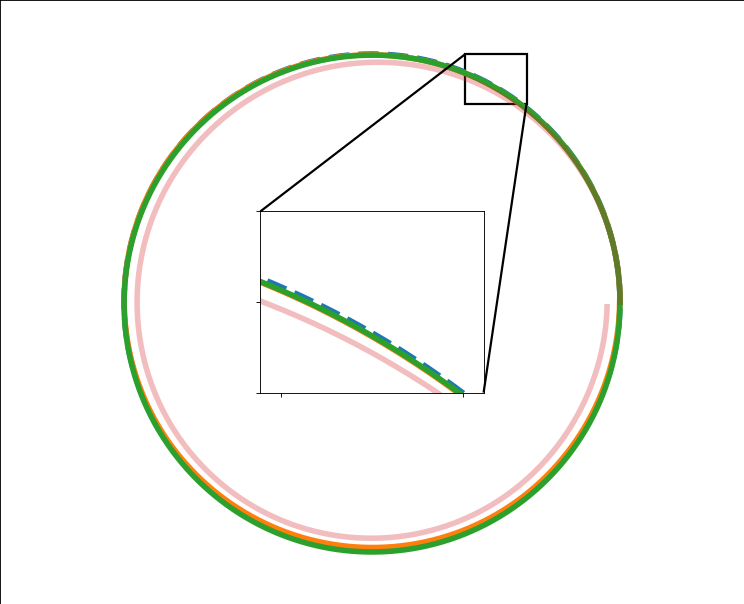}
        \caption{Circular Trajectory} \label{fig: cl-circ}
    \end{subfigure}
    \begin{subfigure}{0.3\textwidth}
        \centering
        \includegraphics[width=\linewidth, trim=5 70 10 5, clip]{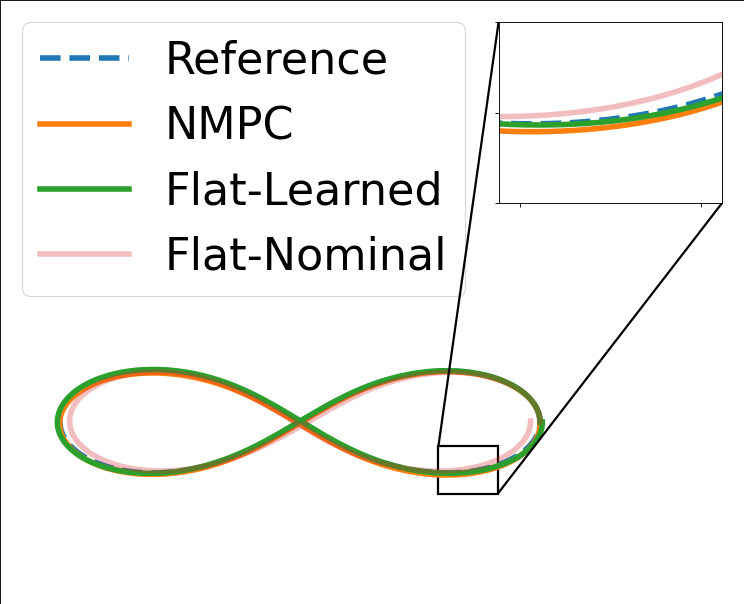}
        \caption{Lemniscate Trajectory} \label{fig: cl-lem}
    \end{subfigure}
    \hfill
    \begin{subfigure}{0.37\textwidth}
        \centering
        \includegraphics[width=\linewidth]{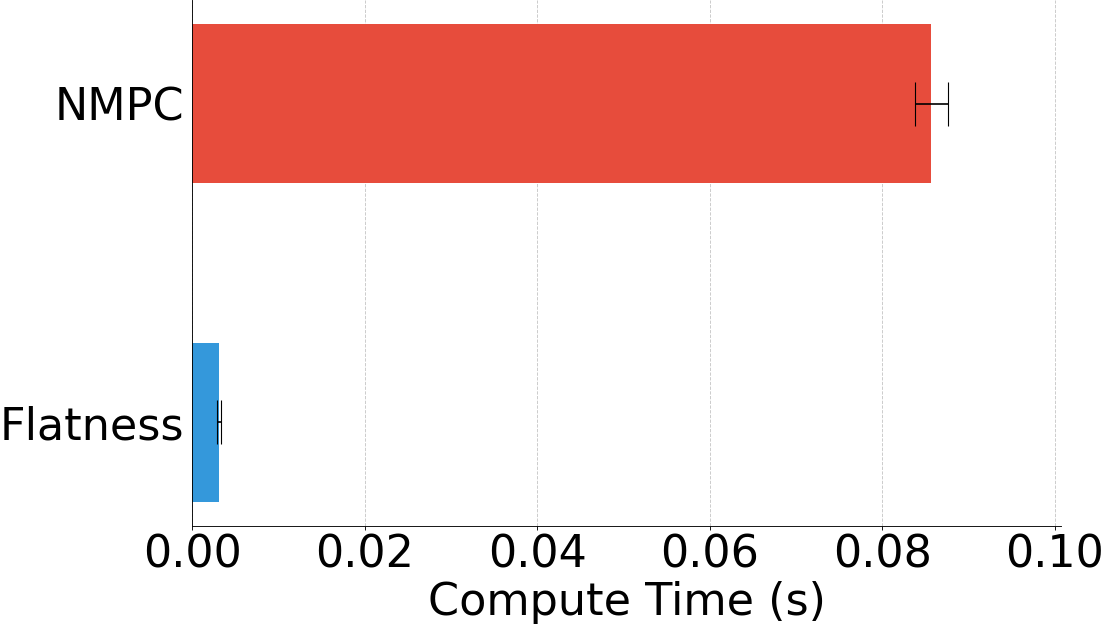}
        \caption{Average Compute Time Comparison} \label{fig: compute-time}
    \end{subfigure}
    \caption{Closed-Loop Tracking Control Results: Figs.~\ref{fig: cl-circ} and \ref{fig: cl-lem} show representative trajectories tracking the circular and lemniscate references, and Fig.~\ref{fig: compute-time} shows mean and standard deviation of compute time per control update for both methods. The flatness-based controller with learned residual dynamics achieves similar performance to nonlinear MPC while requiring $20\times$ less time to compute control actions.}
    \label{fig: closed-loop-plot}
\end{figure*}

\begin{table}[ht]
    \centering
    \begin{tabular}{|c|c|c|c|}
    \hline 
    & \makecell{\textbf{Flatness} \\ \textbf{(Nominal)}} &   
    \makecell{\textbf{Flatness} \\ \textbf{(Learned)}} & \textbf{NMPC} \\ \hline
    \textbf{Circle} & 0.0478  &  0.0051 $\pm$ 0.0009   &   0.0062 $\pm$ 0.0001 \\ \hline
    \textbf{Lemniscate} & 0.0258 &  0.0061 $\pm$ 0.0008   &   0.0063 $\pm$ 0.0001 \\ \hline
    \end{tabular}
    \caption{Mean Closed-loop Tracking Error (m) $\pm$ Standard Deviation over $30$ training runs.}
    \label{tab:closed-loop-result}
\end{table}

For both methods, the control actions sampled with zero-order hold are updated at a rate of $100Hz$. Like the open-loop case, we report the position error in Table~\ref{tab:closed-loop-result} and show representative trajectories in Figure~\ref{fig: closed-loop-plot}. We observe that \rewrite{our learned flatness-based controller is able to perform comparably to NMPC, with both showing more than $5\times$ less position tracking error than the nominal flatness-based controller}{by incorporating the learned residual model, both the flatness-based controller and the NMPC showed more than $5\times$ less position tracking error than the nominal flatness-based controller}. However, despite the need to compute Jacobians of the neural networks, the simple structure of the flatness-based controller allows it to be much more computationally efficient. We show the results on computation time in Figure~\ref{fig: compute-time}. Specifically, the flatness-based controller averages $3.3 \pm 0.3ms$ of wall-clock time per control action update. In contrast, even with warm-starting, NMPC takes $85.2 \pm 1.74ms$. This suggests the potential to run the learned flatness-based controller at a significantly higher frequency or on much more computationally constrained hardware. Although our implementation for both methods has room for optimization, the observed computational savings are in line with those observed in prior work~\citep{sun2022comparative}, wherein the authors show that flatness-based controllers are able to run 2 orders of magnitude faster than NMPC in real-world hardware experiments, with similar tracking performance.


\section{Conclusion}\label{sec:conclusion}
We showed that residual learning can be made compatible with differential flatness-based control methods.  In particular, we showed that augmenting a nominal differentially flat system in pure-feedback form with a lower-triangular residual dynamics term preserves differential flatness of the overall system.  We further showed how to construct the flat maps for the augmented system using those of the nominal system.  Leveraging these results, we proposed an algorithm for learning flatness-preserving residual dynamics for pure-feedback form systems. We operationalized these insights using standard deep learning libraries, and demonstrated the promise of the approach through simulation experiments on 2D quadrotors. This provides a first step towards more principled methods for learning residual models for differentially flat systems. Future work includes relaxing the structural assumptions, extending the algorithm to allow online adaptation to slowly varying residual dynamics, and deploying the algorithm on quadrotors in physical experiments.


\ifniceformat
\bibliographystyle{unsrtnat}
\else
\bibliographystyle{bibFiles/IEEEbib}
\fi
\bibliography{bibFiles/sample}

\appendix
\ifniceformat
\section{Proof of Theorem~\ref{thm: regularity-unperturbed}}%
\else
\subsection{Proof of Theorem~\ref{thm: regularity-unperturbed}}%
\fi
\label{sec: thm1-proof}
We first prove a helpful lemma before formalizing the proof of Theorem~\ref{thm: regularity-unperturbed}. In what follows, we abuse notation and set $\vcx_{r+1}:= \vcu$ to simplify the presentation.

\begin{lemma}\label{lem: implicit-function-lemma}
Under the assumptions of Theorem~\ref{thm: regularity-unperturbed} and 
for each $k = 1, \ldots, r$, there exists an open set ${\mathcal{N}_k \ni (\vcx^*,\vcu^*)}$ and a unique, continuously differentiable function $h_k$ such that for all $(\vcx, \vcu) \in \mathcal{N}_k$,
\begin{equation}\label{eq: implicit-function-lemma}
    \begin{aligned}
        &h_k\left(\vcx_1, \ldots, \vcx_{k}, \bar{f}_k(\vcx_1, \ldots, \vcx_{k+1})\right) = \vcx_{k+1}.
    \end{aligned}
\end{equation}
\end{lemma}
\begin{proof}
\label{sec: lem-proof}
Consider the function ${F_k: \R^{(k+2)m} \to \R^m}$ given by
\begin{equation*}
    F_k(\vcx_1, \ldots, \vcx_{k}, \vcx_{k+1}, \vcc) := \bar{f}_k(\vcx_1, \ldots, \vcx_{k}, \vcx_{k+1}) - \vcc.
\end{equation*}
Clearly, when $\vcc^* = \dot\vcx^*_k := \bar{f}_k(\vcx^*_1, \ldots, \vcx^*_{k+1})$,
\begin{equation*}
    F_k(\vcx^*_1, \ldots, \vcx^*_{k}, \vcx^*_{k+1}, \vcc^*) = 0.
\end{equation*}
By the regularity of $\bar{f}_k$,
\ifniceformat
\begin{equation*}
    \left\lvert D_{\vcx_{k+1}} F_k(\vcx^*_1, \ldots, \vcx^*_{k+1}, \vcc) \right\rvert
    = \left\lvert D_{\vcx_{k+1}} \bar{f}_k(\vcx^*_1, \ldots, \vcx^*_{k+1}) \right\rvert \neq 0
\end{equation*}
\else
\begin{multline*}
    \left\lvert D_{\vcx_{k+1}} F_k(\vcx^*_1, \ldots, \vcx^*_{k+1}, \vcc) \right\rvert  \\
    = \left\lvert D_{\vcx_{k+1}} \bar{f}_k(\vcx^*_1, \ldots, \vcx^*_{k+1}) \right\rvert \neq 0
\end{multline*}
\fi
for all $\vcc\in\R^m$. Thus, by the implicit function theorem, there exists an open neighborhood $V_k \ni (\vcx^*_{1},\ldots, \vcx^*_k, \vcc^*)$ and a unique, continuously differentiable function ${h_k: V_k \to \R^m}$ on which the desired equality \eqref{eq: implicit-function-lemma} holds. Defining
\begin{equation*}
    \mathcal{N}_k := \{\vcx \in \R^{m(r+1)}| \big(\vcx_1, \ldots, \vcx_k, \bar{f}_k(\vcx_1, \ldots, \vcx_{k+1}) \big) \in V_k\},
\end{equation*}
we obtain the desired result.
\end{proof}
With this lemma in hand, we now formalize the proof.

\begin{proof}[Proof of Theorem~\ref{thm: regularity-unperturbed}]

We prove the claim for a neighborhood of $(\vcx^*, \vcu^*)$ of the form
$\mathcal{N} := \bigcap_{i=1}^{r}\mathcal{N}_k$, on which \eqref{eq: implicit-function-lemma} thus holds for all $k = 1, \ldots, r$. 
Lemma~\ref{lem: implicit-function-lemma} thus implies that in this neighborhood, the values of $\vcx_1, \ldots, \vcx_k$ and $\dot\vcx_{k}$ uniquely determine $\vcx_{k+1}$ via the dynamic feasibility constraints. To show that the system \eqref{eq: flat-PFF} is differentially flat, we leverage this fact and explicitly construct the diffeomorphism ${(\Phi, \Psi)}$. 

In particular, we prove the claim by induction on $k$. For the base case $k=1$, note that $\vcx_1 = \Phi_1(\vcy) = \vcy$ for all $(\vcx, \vcu) \in \mathcal{N}$. For the purposes of induction, assume that for all integers $k' < k$, there exists a function $\Phi_{k'}$ such that
\begin{equation*}
    \vcx_{k'} = \Phi_{k'}(\vcy, \dot\vcy, \ldots, \vcy^{(k'-1)}),
\end{equation*}
for $(\vcx, \vcu) \in \mathcal{N}$. Then, by Lemma~\ref{lem: implicit-function-lemma}, we can also express $\vcx_k$ in terms of the flat outputs as
\ifniceformat
\begin{align*}
    \vcx_{k} &= h_{k-1}(\vcx_1, \ldots, \vcx_{k-1}, \dot\vcx_{k-1}) \\
    &= h_{k-1}\left(\Phi_1(\vcy), \ldots, \Phi_{k-1}(\vcy, \ldots, \vcy^{(k-2)}), \frac{d}{dt}\Phi_{k-1}(\vcy, \ldots, \vcy^{(k-2)})\right)\\
    &=:\Phi_k(\vcy, \ldots, \vcy^{(k-1)}),
\end{align*}
\else
\begin{align*}
    \vcx_{k} &= h_{k-1}(\vcx_1, \ldots, \vcx_{k-1}, \dot\vcx_{k-1}) \\
    &= h_{k-1}(\Phi_1(\vcy), \ldots, \Phi_{k-1}(\vcy, \ldots, \vcy^{(k-2)}),\\
    &\qquad\qquad \frac{d}{dt}\Phi_{k-1}(\vcy, \ldots, \vcy^{(k-2)}))\\
    &=:\Phi_k(\vcy, \ldots, \vcy^{(k-1)}),
\end{align*}
\fi
where ${\dot\Phi_{k-1}(\vcy, \ldots, \vcy^{(k-2)})}$ can be computed via
\begin{equation*}
    \frac{d}{dt}\Phi_{k-1}(\vcy, \ldots, \vcy^{(k-2)})) = D\Phi_{k-1} \begin{bmatrix}
        \dot\vcy \\ \vdots \\  \vcy^{(k-1)}
    \end{bmatrix}.
\end{equation*}
The overall flatness diffeomorphism is thus given by
\begin{equation*}
    \Phi = [\Phi_1^\top, \ldots, \Phi_r^\top]^\top, \qquad \Psi = \Phi_{r+1}. \qedhere
\end{equation*}
\end{proof}


\ifniceformat
\section{Experiment Details}\label{sec: exp-detail}
\else
\subsection{Experiment Details}\label{sec: exp-detail}
\fi
The parameters for the quadrotor and the disturbances are given in Table~\ref{tab: quad-params}.
\begin{table}[ht]
    \centering
    \begin{tabular}{|c|c|}
        \hline
        \textbf{Parameter} & \textbf{Value} \\ \hline
        $m$ & 1 \\ \hline
        $I$ & 0.1\\ \hline
        $g$ & 9.81 \\ \hline
        $C_p$ & 0.1 \\ \hline
        $C_r$ & 0.01 \\ \hline
    \end{tabular}
    \caption{Quadrotor and Disturbance Parameters for Numerical Experiments}
    \label{tab: quad-params}
\end{table}

For collecting the training data, we sample the initial condition $\vcx(0)$ uniformly from the set
$$\left\{\vcx \in \R^6\,:\, \vcx_{lb} \leq \vcx \leq \vcx_{ub} \right\},$$
where $\vcx_{ub} = -\vcx_{lb} = [1\;1\;0.5\;0.5\;0.05\;0.1]^\top$.
At each sampling step of data collection, the thrust $F$ is sampled uniformly between $8.98$ and $10.98$ and the torque is sampled uniformly between $[-0.5, 0.5]$.

The gains $k_i$ of the linear feedback controller \eqref{eq: flat-controller-gain} are chosen by pole placement of the closed-loop system, leading to the gains
$$k_0 = 3.96,\;k_1 = 12.08,\;k_2 = 13.16,\;k_3 = 6.03.$$



\end{document}